\documentclass[11pt]{article}

\usepackage[utf8]{inputenc}
\usepackage[T1]{fontenc}
\usepackage{lmodern}
\usepackage[margin=1.1in]{geometry}
\usepackage{amsmath,amssymb,amsthm}
\usepackage{booktabs}
\usepackage{graphicx}
\usepackage{float}
\usepackage[font=small,labelfont=bf,labelsep=period]{caption}
\usepackage{microtype}
\usepackage{enumitem}
\usepackage[hidelinks]{hyperref}

\setlist[enumerate]{itemsep=3pt,topsep=5pt,parsep=0pt}
\theoremstyle{plain}
\newtheorem{theorem}{Theorem}
\newtheorem{proposition}{Proposition}

\title{\bfseries Decomposition Buys Integrity, Not Yield}
\author{Rong He}
\date{\today}

\begin{document}
\maketitle
\begin{center}\small
Code, verification scripts, and the findings log accompany this paper as ancillary files.
\end{center}

\begin{abstract}
\noindent
Multi-agent systems split a task across a tree of agents and justify the split with folklore:
smaller contexts, cleaner separation of concerns, parallelism. We ask what the split does to
the quantity such systems exist to produce --- how much of what the leaves discover actually
reaches the root. Model a decomposition as a tree in which an agent handed $b$ items keeps any
one of them with probability $r(b)$. Two results follow. First, if $r(b)=1/b$, every tree
delivers exactly one finding, for every task size and every shape; we verify this to
$2.4\times10^{-15}$ on 20,000 random irregular trees. Second, if $r(b)=Cb^{-\delta}$, then a
depth-$k$ tree over $N$ findings yields $C^{k}N^{1-\delta}$: task size and architecture
separate completely, and the architecture contributes only the factor $C\le1$ per level. Flat
is therefore optimal for yield, and no arrangement of agents escapes the exponent $\delta$. We
measure $\delta$ on 600 production deep-research traces by three independent identifications,
including a per-item hazard with a cluster bootstrap, and get $\delta=0.34$ (95\% CI
$[0.30,\,0.38]$); a research agent that surfaces four times as many sources carries only 2.3
times as many forward. We also measure $C$ itself, at a hop where the item boundaries are
given by the tool rather than by a text heuristic: a search returns $b$ numbered results and
the agent writes one turn. Because $b=1$ occurs 550 times in the corpus, $C=\rho(1)$ is
observed rather than extrapolated, at $0.571$ $[0.527,\,0.615]$ over 16,082 hops. Each
additional tier therefore costs roughly 45\% of the yield. Depth pays for itself on a second
axis: the root context is the only state that persists and the only state that cannot cheaply
forget, and depth cuts its exposure from $N$ items to $N^{1/k}$. Balancing the two gives a
closed form whose optimum is two or three tiers at measured production parameters and rises by
seven tiers across fourteen orders of magnitude of task size. Finally we test whether deployed
systems follow the rule. On 743,819 production tool calls, a discrete-time hazard model using
only pre-decision information finds that delegation does not respond to a filling context
(odds ratio 0.969 per doubling, $[0.954,\,0.985]$) and is instead an opening move. Two naive
versions of that same regression give $+0.65$ and $-0.37$; both are mechanically biased, and
we report them because the difference between them is the result. A third axis reconciles the
model with practice. Production flat agents bill sub-quadratically in their round count --- a
measured exponent of $1.392\pm0.004$, not the 2 that append-only context predicts. A tier also
costs alignment: on 1,012 annotated multi-agent traces one brief in sixteen goes off-target,
so the per-tier penalty is $C\mu=0.536$ rather than $C$, and at equal spend two tiers overtake
flat at 403 findings. Across every parameter we measured the model says 0.7\% to 11.3\% of
production sessions are worth delegating, against 7.8\% that do --- the right order of
magnitude from parameters none of which were fitted to it, and no sharper than that.
\end{abstract}

\section{Introduction}

A multi-agent system is a bet that a task is better done by several agents than by one. The
bet is usually argued on grounds of context: one agent's window fills, so split the work, give
each piece its own window, and summarize upward. The argument is plausible and nearly
universal. It has not been checked.

Checking it requires deciding what decomposition is supposed to deliver. We take the simplest
answer. A task requires some number of atomic findings --- a file read, a fact retrieved, a
test result --- and the system succeeds to the extent that those findings are present when the
final answer is written. Call the expected number that survive to the root the \emph{yield}.
Yield is not the only thing a system should optimize, but it is the thing the context argument
is implicitly about, and it is the thing we can compute.

Under that definition the context argument fails, and it fails for a reason that is short
enough to state here. Suppose an agent given $b$ items retains any one of them with
probability $r(b)$, and suppose $r$ falls off as a power law, $r(b)=Cb^{-\delta}$. A flat
agent absorbs all $N$ findings and yields $N\cdot r(N)=CN^{1-\delta}$. A depth-$k$ tree with
fan-out $b=N^{1/k}$ yields $N\cdot r(b)^{k}=C^{k}N^{1-\delta}$. The two differ by $C^{k-1}$,
and $C$ is the retention of a single item through a single summarization, which cannot exceed
one. Depth costs; it never pays. Narrower contexts do not buy back what they lose, because the
loss exponent $\delta$ applies to the total either way.

This is a negative result about one objective, not about multi-agent systems. The positive
result is that depth is paid for on a different axis entirely, and that the axis is
identifiable. Of all the contexts in a running system exactly one persists: the root's. It
also happens to be the one context that cannot cheaply forget, because deleting an item from a
cached prefix charges for everything behind it \cite{r11}. Errors that enter an ephemeral
agent's context die with it; errors that enter the root's are permanent. A flat root absorbs
$N$ items and $N$ chances to be corrupted. A depth-$k$ root absorbs $N^{1/k}$. That is the
trade, it is sharp, and it has a closed-form optimum.

\subsection*{Contributions}

\begin{enumerate}
\item A conservation law. At $r(b)=1/b$ every decomposition tree yields exactly one finding,
independent of shape and of $N$ (Theorem~\ref{th:cons}, verified exactly).

\item A factorization. $Y=C^{k}N^{1-\delta}$ separates task size from architecture, and
implies that decomposition cannot raise yield (Theorems~\ref{th:fact} and~\ref{th:nfl}).

\item A measurement of both constants. $\delta=0.34$ $[0.30,\,0.38]$ on 600 production
deep-research traces by three identifications that do not share a failure mode, and $C=0.571$
$[0.527,\,0.615]$ read directly at $b=1$ (Section~\ref{sec:measure}), giving $Y=C^{k}N^{0.66}$.

\item A design rule. The optimal fan-out is where a ray from the origin is tangent to the
$(\ln b,\ln\rho(b))$ curve; the optimal depth solves $\varepsilon N^{1/k}\ln N/k^{2}=\ln(1/C)$
and is two or three at measured parameters (Section~\ref{sec:depth}).

\item A third constant. On 1,012 annotated multi-agent traces the probability that one
agent-to-agent brief is understood is $\mu=0.939$ $[0.935,\,0.960]$, identified by a
framework-fixed-effects fit that separates it from each framework's own baseline
(Section~\ref{sec:brief}). The per-tier penalty is then $C\mu=0.536$.

\item A calibrated cost axis. A production flat agent's token bill grows as $N^{1.39}$, not
the $N^{2}$ an append-only context predicts; at equal spend two tiers overtake flat at 403
findings (Section~\ref{sec:token}).

\item A negative empirical result. Production agent systems do not adapt their topology to
load. Delegation is decided in the first rounds of a session and is, if anything, less likely
as the context fills (Section~\ref{sec:prod}).
\end{enumerate}

Everything runs on a laptop CPU against three corpora that were already public.

\section{The Decomposition Tree}

A task requires $N$ atomic findings. They are discovered at the $N$ leaves of a rooted tree
$T$. Every internal node $v$ receives the outputs of its $b_v$ children, writes one summary,
and passes it up. The root writes the final answer.

A finding arriving at a node of fan-out $b$ survives into that node's summary independently
with probability $r(b)$. We call
\begin{equation}\label{eq:rho}
\rho(b)=b\,r(b)
\end{equation}
the \emph{carrying number}: the expected count of input items represented in one summary. The
yield of a tree is
\begin{equation}\label{eq:yield}
Y(T)=\sum_{\ell\in\mathrm{leaves}(T)}\ \prod_{v\in\mathrm{path}(\ell\to\mathrm{root})} r(b_v).
\end{equation}

Two structural facts make the rest of the paper short.

\begin{proposition}[walk representation]\label{pr:walk}
$Y(T)=\mathbb{E}\big[\prod_v \rho(b_v)\big]$, where the expectation is over a root-to-leaf
walk that picks uniformly among children at each step.
\end{proposition}

\begin{proof}
The walk reaches leaf $\ell$ with probability $\prod_v 1/b_v$. Substitute $r=\rho/b$ into
\eqref{eq:yield}.
\end{proof}

\begin{theorem}[conservation]\label{th:cons}
If $r(b)=1/b$ then $Y(T)=1$ for every tree $T$, every $N$, and every irregular shape.
\end{theorem}

\begin{proof}
$\rho\equiv1$ in Proposition~\ref{pr:walk}.
\end{proof}

Theorem~\ref{th:cons} is the reference point for everything that follows. At proportional
crowding --- each item keeping a share of the summary inversely proportional to how many items
compete for it --- architecture is not merely unimportant, it is \emph{exactly} irrelevant. A
thousand agents in any arrangement deliver what one agent delivers. Any claim that a topology
helps is a claim that retention beats $1/b$; any claim that it hurts is a claim that retention
falls short of it. We verified Theorem~\ref{th:cons} on 20,000 random irregular trees with
fan-outs drawn per node; the worst deviation of $Y$ from 1 was $2.4\times10^{-15}$.

\section{The Yield Factorization}

Take $r(b)=Cb^{-\delta}$ over the range of fan-outs a system actually uses, so
$\rho(b)=Cb^{1-\delta}$.

\begin{theorem}[factorization]\label{th:fact}
For a uniform tree of fan-out $b$ and depth $k$ over $N=b^{k}$ findings,
\begin{equation}\label{eq:fact}
Y=\rho(b)^{k}=C^{k}\,N^{1-\delta}.
\end{equation}
The task-size term $N^{1-\delta}$ does not depend on the architecture, and the architecture
enters only as $C^{k}$.
\end{theorem}

\begin{proof}
$\rho(b)^{k}=C^{k}b^{k(1-\delta)}=C^{k}N^{1-\delta}$.
\end{proof}

Verified to $8.8\times10^{-16}$ over 400 combinations of $(C,\delta,b,k)$.

\begin{theorem}[no free lunch]\label{th:nfl}
$C=\rho(1)\le1$, since one item in cannot produce more than one item out. Hence $k=1$
maximizes \eqref{eq:fact}, strictly whenever $C<1$.
\end{theorem}

Summarization is not the only thing a tier costs. A child also has to understand the brief its
parent wrote, and a flat agent writes none. A depth-$k$ tree has $k$ summarization hops but
only $k-1$ agent-to-agent briefs, so writing $\mu$ for the probability that a delegation
produces work aligned with what was asked,
\begin{equation}\label{eq:full}
Y = C^{k}\,\mu^{\,k-1}\,N^{1-\delta}.
\end{equation}
Since $\mu\le1$ this strengthens Theorem~\ref{th:nfl} rather than qualifying it: the per-tier
penalty is $C\mu$, not $C$. Section~\ref{sec:brief} measures $\mu$.

Theorem~\ref{th:nfl} is worth restating in plain terms, because it contradicts the standard
argument for decomposition. Splitting work so that each agent sees a small, clean context does
not recover the information a large context loses. The rot exponent $\delta$ multiplies $\ln N$
no matter how the tree is shaped, because a tree of depth $k$ applies a $b^{-\delta}$ penalty
$k$ times and $b^{k}=N$. What depth adds on top is $k$ independent summarization steps, each of
which can only lose. Yield-motivated decomposition is futile in the same way that
yield-motivated context pruning is \cite{r11}: the natural optimization is the wrong one, and
it is wrong by a constant factor per operation.

\begin{figure}[t]
\centering
\includegraphics[width=0.62\linewidth]{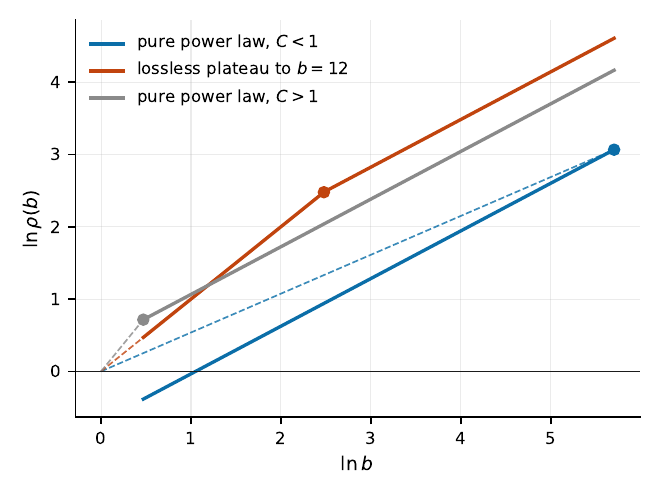}
\caption{The design rule. $\beta(b)=\ln\rho(b)/\ln b$ is the slope of the chord from the
origin, so $b^{\star}$ is the point of tangency. Under a pure power law with $C<1$ the
tangency runs off to the right and flat wins; with $C>1$ it collapses to the smallest $b$;
only a near-lossless plateau produces an interior optimum, at the end of the plateau.}
\label{fig:rule}
\end{figure}

\paragraph{The design rule.}
Theorem~\ref{th:nfl} has an escape hatch. Equation~\eqref{eq:fact} assumed a power law. In
general, comparing architectures at fixed $N$ means maximizing
\begin{equation}\label{eq:beta}
\beta(b)=\frac{\ln\rho(b)}{\ln b},\qquad Y=N^{\beta(b)},
\end{equation}
with ties broken toward the largest $b$ (same yield, fewer tiers, fewer places to inject an
error). Geometrically $\beta(b)$ is the slope of the chord from the origin to
$(\ln b,\ln\rho(b))$, so $b^{\star}$ is where a ray from the origin is tangent to that curve
from above. An interior optimum exists only if $\rho$ has a near-lossless plateau --- a range
of $b$ over which a summarizer keeps essentially everything --- and $b^{\star}$ is the end of
that plateau. Under a pure power law the tangency runs off to the boundary and flat wins,
which is Theorem~\ref{th:nfl} again. Figure~\ref{fig:rule} draws the three regimes.

We checked the rule against four analytic shapes, and it corrected two of our own
expectations. A hard plateau to $b=8$ produces a \emph{tie} in $\beta$ across $b\in[2,\,8]$;
without the largest-$b$ tie-break the rule returns $b^{\star}=2$, which wastes tiers for no
yield. Soft saturation, $\rho(b)=b/(1+b/12)$, gives $b^{\star}=3$ rather than the 12 we
assumed: gradual saturation forces narrow fan-out, and only a sharp plateau permits a wide one.

\section{Measuring the Constants}\label{sec:measure}

$\delta$ is measurable on any agent that is observed choosing among items in its context. We
use 600 production deep-research traces (23,026 searches, 11,889 extraction requests, drawn
from six public question-answering corpora), released with an earlier study of context cost
\cite{r11}. The agent searches, receives about ten results per search, and then calls an
extractor on specific URLs --- 99.6\% of the time on exactly one. Which URLs it pursues is
recorded in the tool arguments, so retention is read off behavior rather than inferred from
text overlap.

\subsection{The crowding exponent}

\paragraph{Aggregate.}
For each trace let $b$ be the number of distinct URLs surfaced and $u$ the number the agent
carried forward. By \eqref{eq:rho}, $u$ estimates $\rho(b)$. A log-log fit over 596 traces
gives
\begin{equation}\label{eq:agg}
u = 0.45\,b^{0.601},\qquad \mathrm{SE}(\beta)=0.034,\quad R^{2}=0.35.
\end{equation}
Figure~\ref{fig:carry} plots the 600 traces against the two degenerate lines. Both degenerate
points are rejected: $\beta=0$, the conservation point of Theorem~\ref{th:cons}, at $z=17.9$;
$\beta=1$, lossless carry, at $z=-11.9$. Median $b$ is 126.5 and median $u$ is 8. Per source
corpus the exponent is 0.633, 0.602, 0.591, 0.517 and 0.390, with standard errors between 0.07
and 0.10.

\begin{figure}[t]
\centering
\includegraphics[width=0.62\linewidth]{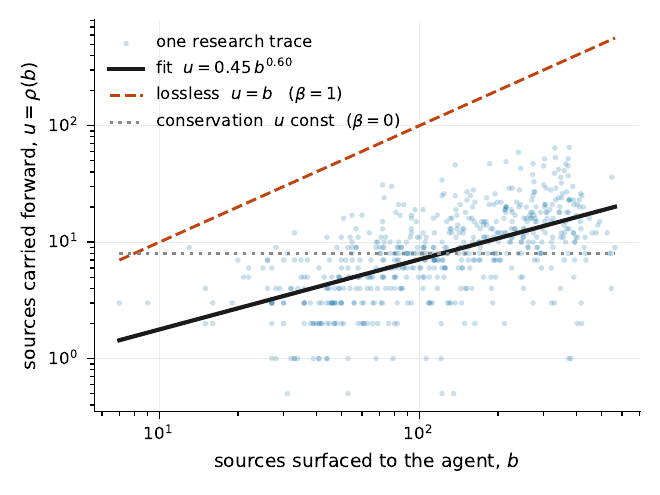}
\caption{Carry-forward on 600 production research traces. Each point is one trace: $b$ sources
surfaced, $u$ carried forward. The fit sits strictly between the two degenerate lines, lossless
carry ($\beta=1$) and the conservation point of Theorem~\ref{th:cons} ($\beta=0$).}
\label{fig:carry}
\end{figure}

\paragraph{The confound.}
Both $b$ and $u$ are the agent's own choices, so a hard question could inflate both. Two
within-trace identifications answer this.

\emph{Instantaneous choice.} At the moment of each extraction the agent picks from the $b$ URLs
already in context. Under uniform choice, the share of picks landing in the ten most recent
URLs would be $10/b$. It is not (Table~\ref{tab:choice}). At $b\approx300$ the agent is six
times more concentrated on recent items than uniform choice allows. Nothing about question
difficulty can produce this, because it is a single decision at a single instant.

\begin{table}[t]
\centering
\caption{Concentration of extraction choices, 8,390 picks.}
\label{tab:choice}
\small
\begin{tabular}{lrrr}
\toprule
$b$ at pick & picks & P(pick in newest 10) & uniform null \\
\midrule
2--20     &   485 & 86.2\% & 83.3\% \\
20--50    & 1,319 & 53.4\% & 29.4\% \\
50--100   & 1,625 & 40.4\% & 13.9\% \\
100--200  & 2,485 & 32.8\% &  6.9\% \\
200--400  & 2,388 & 22.2\% &  3.8\% \\
400+      &    88 & 21.6\% &  2.3\% \\
\bottomrule
\end{tabular}
\end{table}

\emph{Per-item hazard} (Figure~\ref{fig:crowd}). For each of 98,768 surfaced URLs we record the
crowding at its arrival and whether it is ever used, stratified by how many extraction
opportunities remained after it arrived. Within every stratum, $P(\text{used})$ falls
monotonically with crowding (Table~\ref{tab:hazard}). Fitting $\log p = a_s - \delta\log b$
with a stratum intercept and a cluster bootstrap over traces ($B=400$):
\begin{equation}\label{eq:delta}
\delta_{\mathrm{used}}=0.341\ [0.298,\,0.384],\qquad
\delta_{\mathrm{cited}}=0.398\ [0.315,\,0.462].
\end{equation}
The second uses an entirely different outcome --- whether the URL appears in the final answer,
available for 513 traces --- and lands on the same exponent. Three identifications that do not
share a failure mode give $\beta=1-\delta$ of 0.601, 0.659 and 0.602.

\begin{figure}[t]
\centering
\includegraphics[width=0.62\linewidth]{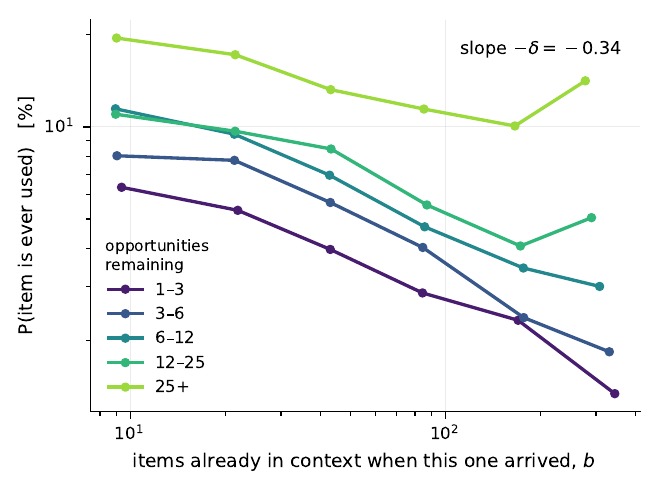}
\caption{Crowding destroys access, and question difficulty does not explain it. Each line holds
the number of remaining extraction opportunities fixed; within every one of them, an item that
arrives into a fuller context is less likely ever to be used.}
\label{fig:crowd}
\end{figure}

\begin{table}[t]
\centering
\caption{$P(\text{URL is ever used})$ by crowding at arrival and remaining opportunity, 98,768
URLs.}
\label{tab:hazard}
\small
\begin{tabular}{lrrrrrr}
\toprule
& \multicolumn{6}{c}{$b$ at arrival} \\
\cmidrule(l){2-7}
opportunities left & 5--15 & 15--30 & 30--60 & 60--120 & 120--250 & 250+ \\
\midrule
1--3   &  6.34\% & 5.33\% & 3.97\% & 2.87\% &  2.33\% &  1.34\% \\
3--6   &  8.04\% & 7.77\% & 5.66\% & 4.03\% &  2.38\% &  1.84\% \\
6--12  & 11.45\% & 9.46\% & 6.95\% & 4.71\% &  3.46\% &  3.01\% \\
12--25 & 10.99\% & 9.67\% & 8.46\% & 5.55\% &  4.08\% &  5.05\% \\
25+    & 19.49\% & 17.19\% & 13.22\% & 11.42\% & 10.06\% & 14.12\% \\
\bottomrule
\end{tabular}
\end{table}

The URL-level design cannot reach $C$. It is an extrapolation to $b=1$, outside the observed
range $b\in[5,\,564]$, and at small $b$ it is demand-limited: the agent stops extracting because
it has enough, not because it cannot see. $\delta$ survives that because demand-limiting is
absorbed into the stratum intercepts; $C$ does not.

\subsection{The per-level constant}

$C$ needs a hop where a single item can be observed entering and either surviving or not. The
corpus contains one. \texttt{web\_search} returns a \emph{numbered} list, so item boundaries
come from the tool rather than from a text heuristic, and the number of results varies from 1
to 50 across the corpus. Crucially $b=1$ occurs 550 times, so $\rho(1)$ is observed rather than
extrapolated. The hop is: $b$ results enter the context, the agent writes one turn. An item
counts as carried if that turn --- reasoning, prose, or the arguments of its next tool call ---
names the result's URL or the distinctive words of its title. A permutation null matches the
same items against a different hop's turn.

Over 16,082 hops, raw retention is 0.358 and the permutation null is 0.0012, so the matching is
specific. Figure~\ref{fig:const} shows the curve and where $C$ is read off; the right panel
shows that $r(b)$ sits above $1/b$ at every $b$ we observe, so $\beta>0$ throughout.

\begin{figure}[t]
\centering
\includegraphics[width=\linewidth]{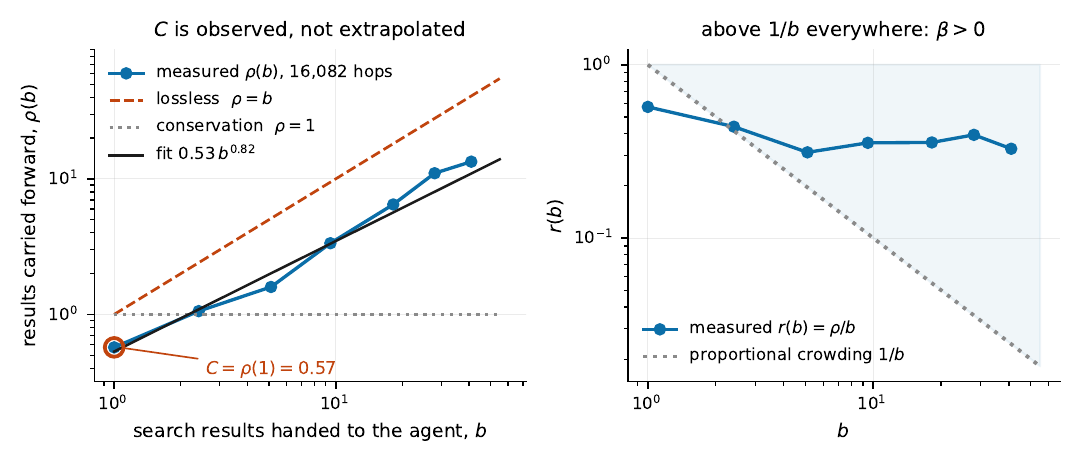}
\caption{The per-level constant, observed rather than extrapolated. \emph{Left:} carry-forward
against batch size; $b=1$ occurs 550 times, so $C=\rho(1)$ is a measurement. \emph{Right:} the
same data as a retention rate, which lies above proportional crowding at every observed $b$.}
\label{fig:const}
\end{figure}

\begin{table}[t]
\centering
\caption{Carry-forward at a single search hop, chance-corrected.}
\label{tab:const}
\small
\begin{tabular}{rrrrr}
\toprule
$b$ & hops & $\rho(b)$ & $r(b)=\rho/b$ & $1/b$ \\
\midrule
 1 &   550 &  0.571 & 0.571 & 1.000 \\
 2 &   364 &  1.008 & 0.504 & 0.500 \\
 3 &   251 &  1.127 & 0.376 & 0.333 \\
 5 &   236 &  1.538 & 0.308 & 0.200 \\
 7 &   480 &  2.752 & 0.393 & 0.143 \\
10 & 7,910 &  3.540 & 0.354 & 0.100 \\
$\approx$18 &   612 &  6.477 & 0.356 & 0.055 \\
$\approx$28 &   404 & 11.012 & 0.394 & 0.036 \\
$\approx$41 &   176 & 13.403 & 0.328 & 0.024 \\
\bottomrule
\end{tabular}
\end{table}

\begin{equation}\label{eq:C}
C=\rho(1)=0.571\ [0.527,\,0.615].
\end{equation}
Fits agree with the direct reading: $0.532$ $[0.500,\,0.562]$ over $b\in[1,\,10]$, $0.469$ over
$b\in[2,\,10]$ with $b=1$ held out, $0.475$ over the full range to $b=50$. Every variant rejects
$C\ge1$. Charging summarization alone, a two-tier system keeps 57.1\% of the findings a flat
agent keeps, three tiers 32.6\% and four tiers 18.6\%; Section~\ref{sec:brief} adds the second
per-tier cost and revises these downward.

Table~\ref{tab:const} also shows something we did not go looking for. From $b=5$ to $b=41$ the
retention rate $r(b)$ is essentially flat at about one third, an exponent of 0.13 rather than
0.34: the agent names about a third of whatever batch is put in front of it, however large the
batch. Yet the probability that an item is \emph{ever} acted on decays with $\delta=0.341$.
Information is not lost at the moment it arrives; it is lost while the context fills. The
exponent that belongs in $N^{1-\delta}$ is the second one, measured over the root's whole
accumulation.

\subsection{The briefing loss}\label{sec:brief}

$\mu$ is a property of the handoff, not of the summary, so it needs a corpus of multi-agent
traces with labelled coordination failures. MAST-Data \cite{r2} is one: 1,642 traces from seven
frameworks, eight benchmarks and five models, each carrying binary annotations for fourteen
failure modes, released under CC-BY-4.0. Category 2 of that taxonomy, \emph{inter-agent
misalignment}, is the authors' own grouping for a delegation that goes off-target, so we take
the whole category rather than picking modes from it. It fires in 53.3\% of traces; category 1
fires in 58.1\%, category 3 in 55.2\%, and 75.3\% of traces carry at least one mode.

MAST annotates per trace, and $\mu$ is per brief. Converting between them requires the number
of agent-to-agent handoffs in each trajectory, which means parsing seven log formats. Two of
our first counters were wrong in ways that would have dominated the answer: MetaGPT came out at
a median of one handoff because the pattern matched only the opening human message, and AG2 at
zero. Left uncorrected, MetaGPT's 63\% failure rate over $h=1$ would have forced $\mu=0.37$.
With corrected counters the usable sample is 1,012 traces --- 1,015 have a countable handoff
structure and three of those record none, leaving the model undefined for them
(Table~\ref{tab:frameworks}).

\begin{table}[t]
\centering
\caption{Handoff structure and misalignment rate by framework.}
\label{tab:frameworks}
\small
\begin{tabular}{lrcr}
\toprule
framework & traces & handoffs: median / p90 / max & P(category 2) \\
\midrule
MetaGPT    & 429 &  4 / \phantom{0}4 / \phantom{0}4 & 63.2\% \\
ChatDev    & 330 & 14 / 20 / 23 & 62.1\% \\
Magentic   & 195 & 17 / 39 / 41 & 60.0\% \\
AppWorld   &  30 &  8 / 27 / 91 & 76.7\% \\
HyperAgent &  28 &  4 / 12 / 40 & 82.1\% \\
\bottomrule
\end{tabular}
\end{table}

AG2 (597 traces) is excluded because its serialized trajectories expose no speaker transitions
in 95\% of cases, and OpenManus (30) because it is a single-agent planning flow with no
agent-to-agent brief.

\paragraph{Falsification test.}
The functional form is falsifiable, so we test it before using it. Modelling misalignment as
independent per brief gives $P(\text{fires}\mid h)=1-(1-q)^{h}$, which requires the failure
rate to rise with the handoff count. Pooled, the logistic slope on $\log h$ is
$+0.308\pm0.086$ ($z=3.6$). Within framework, which removes format and task-mix differences,
Magentic gives $+1.842\pm0.269$ ($z=6.8$) and ChatDev $+1.185\pm0.730$ ($z=1.6$); the other
three have almost no variation in $h$ and are uninformative about the slope. The form survives,
but on the evidence of essentially one framework.

That same fact makes the pooled maximum-likelihood estimate unidentified. MetaGPT contributes
429 traces at a near-constant $h=4$, so its $q$ absorbs a framework baseline rather than a
per-brief rate, and the pooled fit returns $\mu=0.900$. Separating the two,
\begin{equation}\label{eq:fe}
P(\text{fires}\mid f,h) = 1-(1-\alpha_f)(1-q)^{h},
\end{equation}
where $\alpha_f$ is framework $f$'s own baseline and $q$ is identified only by variation in $h$
within a framework, gives
\begin{equation}\label{eq:mu}
\mu = 1-q = 0.939\ [0.935,\,0.960].
\end{equation}
The fitted baselines are 0.000 for ChatDev and Magentic --- their failures are accounted for by
handoff count alone --- against 0.525 for MetaGPT, 0.593 for AppWorld and 0.743 for HyperAgent,
exactly the frameworks whose $h$ does not move. Fitting the two informative frameworks on their
own gives 0.933 and 0.946, which brackets the pooled fixed-effects estimate. The narrower mode
set (task derailment and ignored input) gives $\mu=0.977$ and the wider set (category 2 plus
disobeying task and role specification) $\mu=0.944$.

About one brief in sixteen goes off-target. The per-tier yield penalty is therefore
$C\mu=0.536$ rather than $C=0.571$: two tiers keep 53.6\% of what a flat agent keeps, three
tiers 28.7\%, four tiers 15.4\%.

Two biases both push $\mu$ down, so the depth penalty here is an upper bound. A run that goes
wrong loops more and therefore logs more handoffs, which inflates $q$. And $\mu$ multiplies a
whole subtree's yield, treating a misaligned delegation as a total loss, when a misaligned
subagent usually still returns something usable.

\section{Root Integrity and the Optimal Depth}\label{sec:depth}

Exactly one context in a running system persists for the whole task: the root's. It is also the
only one that cannot cheaply forget. Removing an item from a cached prefix costs a re-prefill
of everything behind it, so an agent's context is an append-only ratchet \cite{r11}. An error
written into an ephemeral agent's context is discarded when that agent finishes. An error
written into the root's is permanent.

Let $\varepsilon$ be the probability that an item absorbed into the persistent context is
wrong. A flat root absorbs all $N$ findings; a depth-$k$ root absorbs $N^{1/k}$ summaries:
\begin{equation}\label{eq:integ}
P(\text{root uncorrupted}) = (1-\varepsilon)^{N^{1/k}}.
\end{equation}
This is the whole case for depth, and it is exponential where the yield cost is only
geometric. Combining \eqref{eq:full} and \eqref{eq:integ},
\begin{equation}\label{eq:J}
J(k)=k\ln C+(k-1)\ln\mu+(1-\delta)\ln N+N^{1/k}\ln(1-\varepsilon),
\end{equation}
whose first-order condition is
\begin{equation}\label{eq:foc}
\frac{\varepsilon\,N^{1/k}\ln N}{k^{2}}=\ln\frac{1}{C\mu}.
\end{equation}
The left side falls in $k$, so the optimum is unique. Figure~\ref{fig:trade} shows the two
terms and their product. Equation~\eqref{eq:foc} reproduces the argmax of \eqref{eq:J} within
one tier in 80 of 80 parameter combinations.

\begin{figure}[t]
\centering
\includegraphics[width=\linewidth]{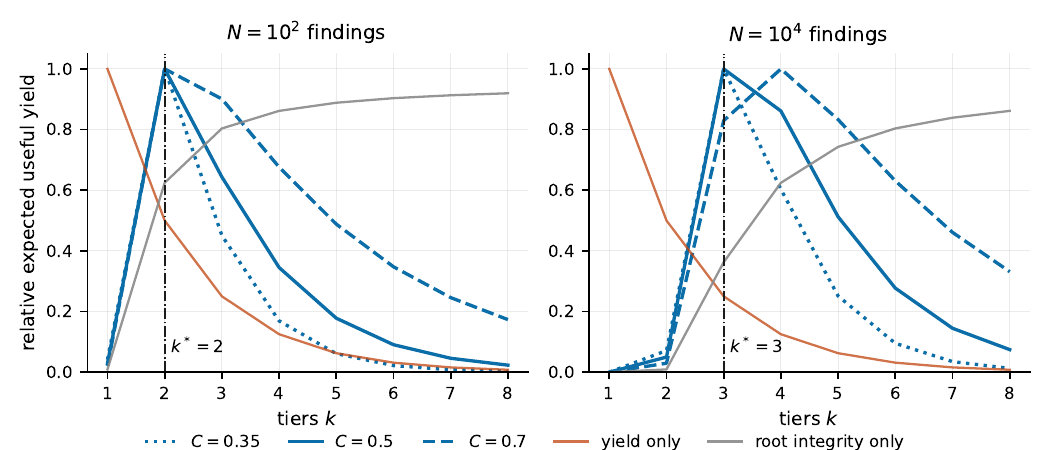}
\caption{The two axes and their product, at two task sizes. Yield falls geometrically in the
number of tiers while root integrity rises exponentially, so the optimum is interior and, at
realistic parameters, small.}
\label{fig:trade}
\end{figure}

Substituting the measured constants, $C\mu=0.536$ and $\delta=0.341$, into \eqref{eq:foc} gives
$k^{\star}=2$ at the median production session and $k^{\star}=3$ at its 99th percentile.
Figure~\ref{fig:kstar} maps $k^{\star}$ over $(N,\varepsilon)$ with the production sessions
marked. Two properties of $k^{\star}$ matter more than its exact value. It is small, and it
stays small (Table~\ref{tab:kstar}).

\begin{figure}[t]
\centering
\includegraphics[width=\linewidth]{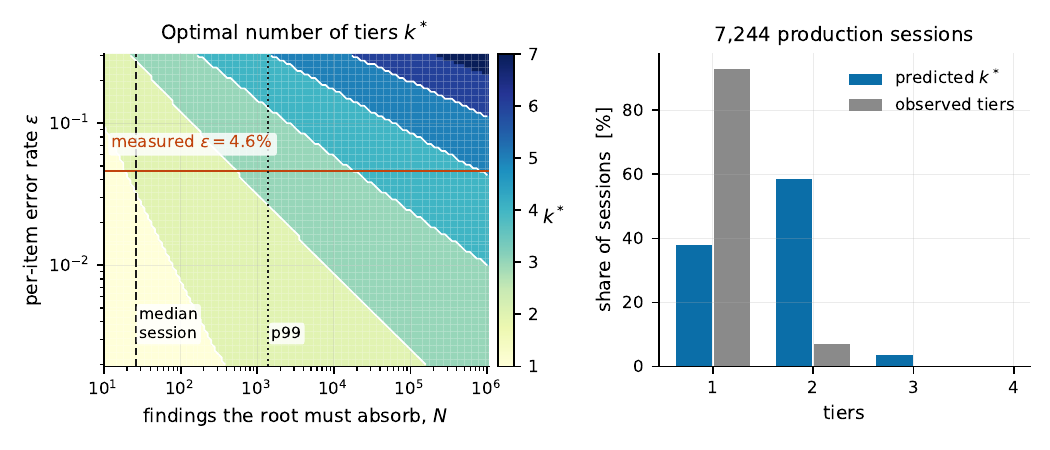}
\caption{\emph{Left:} the optimal number of tiers over task size and per-item error rate, with
the measured production error rate and the median and 99th-percentile session marked.
\emph{Right:} predicted against observed tiers across 7,244 production sessions.}
\label{fig:kstar}
\end{figure}

\begin{table}[t]
\centering
\caption{The optimal number of tiers is small and stays small. $\delta=0.34$,
$\varepsilon=0.02$.}
\label{tab:kstar}
\small
\begin{tabular}{lrrrrrrrr}
\toprule
$N$ & $10^{2}$ & $10^{4}$ & $10^{6}$ & $10^{8}$ & $10^{10}$ & $10^{12}$ & $10^{14}$ & $10^{16}$ \\
\midrule
$k^{\star}$ & 2 & 3 & 4 & 5 & 6 & 7 & 8 & 9 \\
\bottomrule
\end{tabular}
\end{table}

Fourteen orders of magnitude of task size move the optimum by seven tiers. Deep agent
hierarchies are not a thing that large tasks should grow into. The table is identical at
$C\mu=0.536$, at $C=0.571$ and at the placeholder $0.5$ we used before either was measured:
$k^{\star}$ is insensitive to the per-tier penalty over that range. It is robust to $C\mu$ over
$[0.2,\,0.7]$ and to $\varepsilon$ over $[0.002,\,0.2]$; only near-lossless summarization
($C\mu\ge0.9$) pushes past four tiers at realistic $N$.

The biological analogue is exact on this axis and instructive where it fails. Renewing tissues
are organized as stem cells, transit-amplifying cells, and differentiated cells --- two or
three tiers --- and the standard explanation is the same one as \eqref{eq:integ}: only the stem
compartment persists, so architecture exists to keep the number of divisions in the persistent
compartment small \cite{r10}. Where the analogy breaks is the niche itself. A stem cell niche
exists because cells cannot de-differentiate. An agent can: its specialization is a
\emph{suffix} of its context, and suffix truncation is the one free edit \cite{r11}. Resetting
a specialized agent to its base prompt costs nothing beyond spawning a fresh one, so there is
no reason to maintain a reserve of unspecialized agents. Asymmetric division and bounded depth
carry over; the niche does not.

\section{The Token Axis}\label{sec:token}

Yield and integrity are not the only things a decomposition changes. It also changes the bill,
and the bill is what most practitioners are actually managing. Including it is not optional:
adding only the \emph{cost of delegating} would stack the deck, because that cost penalizes
depth, while the prefix-caching \emph{saving} from keeping each context short rewards it. Both
belong in the same model.

\paragraph{The measured cost exponent.}
The textbook argument is wrong, and the corpus says so. A flat agent's context is append-only,
so with prefix caching each round re-reads the whole prefix and the bill should be quadratic in
the number of findings absorbed. Regressing the log of a session's total input tokens on the
log of its round count over 6,567 production sessions gives an exponent of $1.392\pm0.004$
($R^{2}=0.939$); separately, $1.450\pm0.005$ for Claude Code and $1.337\pm0.007$ for Codex.
Append-only predicts 2 and a constant-size context predicts 1. Production flat agents are
already sub-quadratic, because they compact.

That observation suggests context collapse is decomposition along the time axis rather than the
agent axis, which in turn suggests compaction and delegation should be substitutes. They are
not. Holding task size fixed, sessions that compact delegate \emph{more}: 20.7\% against 8.7\%
for sessions of 40 to 160 tool results, 29.5\% against 27.0\% above 160, and no difference at
all below 40. Systems under context pressure reach for both.

\paragraph{A calibrated cost model.}
Let $a$ be the measured exponent, $v$ the mean size of a tool result, and $c$ the overhead of
one delegation --- the brief out plus the summary back. For a uniform tree of fan-out
$b=N^{1/k}$ with $M=(N-1)/(b-1)$ internal nodes,
\begin{equation}\label{eq:cost}
T(N,k) = M\,v\,b^{a} + (M-1)\,c.
\end{equation}
Measured on the corpus, $v=4{,}194$ characters and $c=7{,}694$ characters (2,640 of brief,
5,054 returned), from 2,546 \texttt{Agent} calls.

\paragraph{The budget-matched question.}
At a fixed budget $B$, a flat agent can afford $N_1$ findings and a two-tier system can afford
a larger $N_2$, because its bill grows more slowly. The architecture that delivers more is the
one maximizing
\begin{equation}\label{eq:budget}
Y\big(N_k,k\big)\quad\text{subject to}\quad T(N_k,k)=B.
\end{equation}
Solving both sides at equal spend, with the briefing loss of Section~\ref{sec:brief} included,
gives a crossover at $N^{\star}=403$ findings; without it, 246 (Figure~\ref{fig:token}). Below
it a flat agent returns more findings per token, above it two tiers do. The crossover is where
the yield penalty $C\mu$ is finally outweighed by being able to afford more work.

\begin{figure}[t]
\centering
\includegraphics[width=\linewidth]{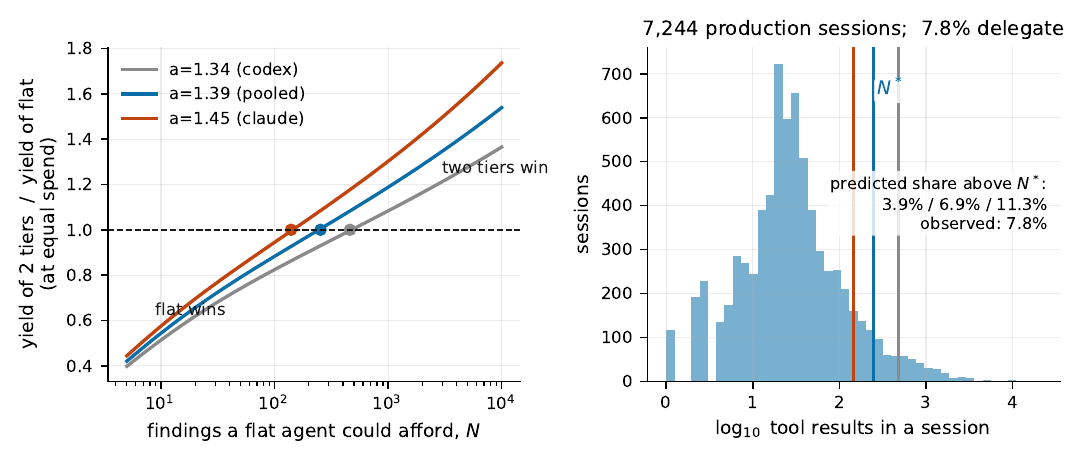}
\caption{\emph{Left:} yield of two tiers relative to flat at equal spend, at each provider's
own measured cost exponent; the crossover is where decomposition starts paying.
\emph{Right:} the size distribution of 7,244 production sessions against those thresholds.}
\label{fig:token}
\end{figure}

\section{What Production Systems Actually Do}\label{sec:prod}

We take production parameters from 743,819 tool calls across 8,058 Claude Code and Codex
sessions in the same release \cite{r11}.

\paragraph{Measured inputs.}
Tool calls return an error 4.60\% of the time (Claude) and 4.88\% (Codex). Sessions absorb a
median of 22 (Claude) and 40 (Codex) tool results, with 99th percentiles of 815 and 2,618. When
a context collapses, the surviving context is a median 15.2\% of the previous one by tokens
(IQR 9.6--26.4\%, 4,892 events) --- an upper bound on what any summarization step can carry.

\paragraph{Measured topology.}
Of 7,244 sessions with at least one tool call, 563 (7.8\%) ever spawn a subagent. Among rounds
that spawn, 74.5\% spawn exactly one; the mean parallel fan-out is 1.49; no deeper nesting
appears in the corpus. Production agent swarms are, empirically, almost flat.

\paragraph{The load-adaptation test.}
Equation~\eqref{eq:foc} predicts $k^{\star}$ rises with $N$, so an adaptive system should
delegate more as its context fills. Testing this is harder than it looks, and we report three
attempts because the difference between them is the finding.

Regressing whether a session delegates on the log of its total tool calls gives a slope of
$+0.649\pm0.031$. This is an artifact: a session that delegates records its subagent's tool
calls in its own ledger, so delegating inflates $N$. Recounting $N$ using only calls up to the
first spawn flips the slope to $-0.372\pm0.033$. This is also an artifact, in the opposite
direction: truncating at the first spawn assigns $N=3$ to a session that delegated at round 3.
Neither number identifies anything.

The identified design is a discrete-time hazard model that uses only information available
before the decision. The risk set is every session-round up to and including the first spawn
(502,718 rounds, 563 events). The outcome is whether the first spawn occurs at that round. The
covariate is accumulated tool-result volume \emph{before} that round, controlling for rounds
elapsed, which separates a full context from a long-running session
(Table~\ref{tab:hazardfit}).

\begin{table}[t]
\centering
\caption{Discrete-time hazard of the first delegation. 502,718 session-rounds, 563 events.}
\label{tab:hazardfit}
\small
\begin{tabular}{lrrr}
\toprule
term & coefficient & SE & $z$ \\
\midrule
log accumulated result chars & $-0.0448$ & 0.0117 & $-3.8$ \\
log rounds elapsed           & $-0.5864$ & 0.0326 & $-18.0$ \\
\bottomrule
\end{tabular}
\end{table}

Doubling the accumulated context multiplies the odds of delegating now by 0.969
$[0.954,\,0.985]$. With tool count in place of volume the coefficient is not significant
($-0.079\pm0.053$). Elapsed rounds dominates and is strongly negative.

Delegation in these systems is an opening move. Topology is chosen in the first rounds and
never revised, which is the standard complaint about fixed-role multi-agent frameworks turned
into a measurement. One caveat is required: unobserved heterogeneity biases duration dependence
downward in any single-spell hazard model, because sessions that ran long without delegating
are disproportionately sessions that never would. This is evidence against strong load
adaptation, not proof that none exists.

\paragraph{Reconciling the model with the topology.}
Section~\ref{sec:depth}'s objective --- yield times integrity, with no cost term --- says
70.1\% of these sessions should run two or more tiers. The observed figure is 7.8\%, and
measuring $C$ made that worse rather than better, because a larger $C$ makes hierarchy cheaper.

The token axis closes most of the gap. The share of the 7,244 sessions whose $N$ exceeds the
crossover $N^{\star}$ is given in Table~\ref{tab:share}, ignoring the briefing loss. The
observed rate sits inside the range the two providers' own exponents produce. Nothing here was
fitted to it: $C$ and $\delta$ come from the research traces, $a$ from session token bills, $v$
and $c$ from tool-call statistics.

\begin{table}[t]
\centering
\caption{Predicted against observed delegation, before the briefing loss is charged.}
\label{tab:share}
\small
\begin{tabular}{lrrrr}
\toprule
exponent $a$ & measured on & $N^{\star}$ & predicted share & observed \\
\midrule
1.337 & Codex sessions       & 477 & \phantom{0}3.9\% & \\
1.392 & both, pooled         & 246 & \textbf{6.9\%}   & \textbf{7.8\%} \\
1.450 & Claude Code sessions & 145 & 11.3\%           & \\
\bottomrule
\end{tabular}
\end{table}

\paragraph{Effect of the briefing loss.}
It moves the prediction the wrong way, and we report it anyway. $\mu$ penalizes depth, so
adding it pushes $N^{\star}$ up and the predicted share down: to 403 and 4.5\% at the
identified $\mu=0.939$, 294 and 5.8\% at the narrow mode set, 570 and 3.3\% at the unidentified
pooled estimate (Figure~\ref{fig:brief}). Sweeping everything we measured --- $\mu$ from 1.0 to
0.944, $a$ from 1.337 to 1.450 --- gives a predicted share between 0.7\% and 11.3\%. The
observed 7.8\% lies inside that interval, but the interval is wide. The 6.9\%-against-7.8\%
agreement in Table~\ref{tab:share} is sharper than the evidence supports, and the honest
summary is that the token axis brings a model that over-predicted delegation nine-fold into the
right order of magnitude, not that it predicts the rate.

\begin{figure}[t]
\centering
\includegraphics[width=\linewidth]{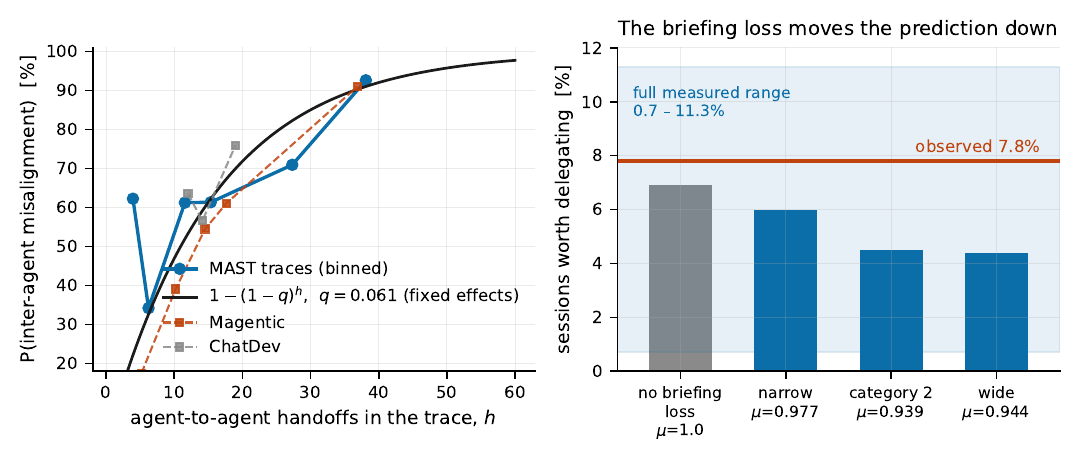}
\caption{\emph{Left:} misalignment against handoff count on 1,012 annotated multi-agent traces,
with the fitted independent-per-brief curve; the two frameworks whose handoff count varies
enough to identify the slope are shown separately. \emph{Right:} charging the briefing loss
moves the predicted share of sessions worth delegating away from the observed rate, not toward
it.}
\label{fig:brief}
\end{figure}

\paragraph{The integrity term.}
With the token axis carrying the explanation there is little room left for integrity.
Reintroducing $(1-\varepsilon)^{N^{1/k}}$ into \eqref{eq:budget} moves $N^{\star}$ from 246 at
$\varepsilon=0$ to 212 at $\varepsilon=10^{-4}$, 117 at $10^{-3}$, and 15 at the measured
tool-failure rate of 4.6\%, where the model predicts near-universal delegation. Matching 7.8\%
requires $\varepsilon$ below about $10^{-4}$.

We do not believe silent contamination is four orders of magnitude rarer than visible tool
failure. The more likely fault is our own integrity term: $(1-\varepsilon)^{N}$ treats any
contamination as total failure, and one wrong fact in a long context usually is not a failed
task. A graded term would sit between. Either reading supports the same conclusion, which is
the one that matters: the rate that belongs in an integrity term is not the rate at which tools
return errors, and an earlier version of this paper inverted for $\varepsilon$ and obtained
0.275\% only because the cost axis was missing and integrity was being asked to do its job.

\paragraph{Rates and mechanisms.}
That production delegates at about the rate the token axis recommends is not evidence that it
is optimizing anything. The hazard model above found no within-session response to a filling
context. Production lands on roughly the right \emph{amount} of delegation while allocating it
without regard to load --- which is exactly the gap an adaptive topology would fill.

\section{Related Work}

Error propagation in multi-agent systems has been studied empirically. MAST \cite{r2} builds a
failure taxonomy from 150 expert-annotated traces and releases 1,600 more, finding that
inter-agent misalignment and missing verification dominate. Hallucination cascade analyses
\cite{r3} measure per-stage amplification and attenuation in short agent chains. Resilience
work \cite{r6} injects faulty agents and finds hierarchical oversight more robust than flat peer
discussion. These establish that errors propagate and that topology matters; none gives a law
relating yield to depth, and none identifies a threshold. Topology studies \cite{r4} find that
\emph{moderately} sparse communication does best, suppressing error propagation while
preserving useful diffusion --- an interior optimum in the same spirit as ours, reached
empirically. The scaling study closest to ours \cite{r5} fits a predictive model across 180
configurations and five architectures and reaches cross-validated $R^{2}=0.373$ over six
benchmarks (0.413 with a task-grounded capability metric), but the model is empirical rather
than structural, and it does not separate task size from architecture.

The mathematics we use is old and, as far as we can tell, has not been applied here. Von
Neumann's threshold theorem for computation with unreliable components \cite{r7} and its sharp
successors for noisy circuits \cite{r8} and for broadcasting on trees \cite{r9} all ask when
information survives a noisy tree. The parallel is closer than it first appears: Evans and
Schulman place criticality where the product of the per-gate information transmission factor
and the fan-in equals one \cite{r8}, and our conservation point is $\rho(b)=b\,r(b)=1$, the
same product set to the same value. We reach it from the opposite structure. In a noisy circuit
the children of a gate are redundant copies, so fan-in supplies error correction and the
question is whether redundancy outruns noise. In an agent tree the children work on
\emph{different} subtasks, so fan-out supplies no redundancy at all and aggregation is
conjunctive. The shared threshold with the opposite mechanism is, we think, the reason this
literature has not been applied here.

Retention under crowding is documented as position bias: accuracy on multi-document question
answering is highest at the ends of a context and degrades in the middle, and the degradation
deepens as the number of documents rises from 10 to 20 to 30 \cite{r1}. That literature
measures $r(b)$ at a few values of $b$ without naming it or asking what it implies for
architecture.

The cost model for context edits, the ratchet property that makes the root context the only
persistent state, and two of the three corpora used here come from an earlier study of context
eviction economics \cite{r11}.

\section{Limitations}

$C$ is measured at one hop --- a search result carried into the agent's next turn --- and
assumed constant across levels. In a real hierarchy the item granularity changes with level: a
subagent's whole report is not the same kind of object as a single tool result. A single
$(C,\delta)$ pair across levels is an idealization, and the two exponents we measure (0.13
immediately, 0.34 eventually) are direct evidence that one power law does not describe the
whole curve.

An item the agent had no use for counts against $C$, so the estimate is conservative in a way
we cannot quantify: it measures what is carried, not what could have been.

The retention model treats items as surviving independently. Real summarizers drop correlated
groups, which will make the effective $b$ smaller than the nominal one and should bias $\delta$
toward zero.

The integrity term treats any contamination as total failure. $(1-\varepsilon)^{N}$ is too
harsh for a long context in which one wrong fact is usually survivable, and
Section~\ref{sec:prod} shows the model is sensitive to exactly this. A graded degradation term
is the obvious next version and we have not built one.

The cost model prices prefill and delegation overhead. It ignores latency and output tokens.

$\mu$ is identified by variation in handoff count within essentially one framework, on traces
from benchmark runs rather than production. Both of its biases point the same way --- a failing
run logs more handoffs, and a misaligned delegation is treated as a total loss --- so 0.939 is
a lower bound on alignment and the depth penalty it implies is an upper bound.

The deep-research corpus is one scaffold with one tool vocabulary. The exponent replicates
across the six question corpora it draws from, but not across agent designs.

The hazard result is subject to frailty bias, as noted in Section~\ref{sec:prod}.

\section{Conclusion}

Decomposition does not create information. Under power-law retention the yield of a
decomposition tree is $C^{k}N^{1-\delta}$: the task-size term is the same for every
architecture, and depth contributes only a factor $C\le1$ per level. Both constants are
measurable on traces that already exist. On 600 production research traces $\delta=0.34$
$[0.30,\,0.38]$ by three identifications that do not share a failure mode, and $C=0.571$
$[0.527,\,0.615]$ read directly at $b=1$ over 16,082 hops. The law for this agent is
$Y=C^{k}\mu^{k-1}N^{0.66}$. With the briefing loss of Section~\ref{sec:brief} charged, one tier
of hierarchy costs 46.4\% of the findings. None of the three constants is a property of the
tree; all three are properties of the agent.

What depth buys lies on two other axes. The root context is the only state that persists and
the only state that cannot cheaply forget, and depth cuts its exposure from $N$ items to
$N^{1/k}$. Depth is also cheaper: a production flat agent's bill grows as $N^{1.39}$, and at
equal spend two tiers overtake flat at 403 findings once the briefing loss is charged. Across
every parameter we measured the model puts 0.7\% to 11.3\% of production sessions above that
threshold, against 7.8\% that delegate: the right order of magnitude, and no sharper. The rule
is simple enough to apply: fan out as wide as summarization stays near-lossless and no wider,
and add a tier when the task is large enough that the bill, not the information, is the binding
constraint.

Deployed systems arrive at roughly the right amount of delegation without appearing to compute
it. They choose a topology in the first rounds of a session and hold it, and the hazard of
delegating is flat or slightly falling in how full the context has become. The amount is about
right; the allocation ignores load. Making topology respond to load is a small change to an
agent harness and, on these numbers, the one with the clearest justification.

\section*{Reproducibility}

Every number in this paper is produced by a script in \texttt{anc/code/}, on a laptop CPU,
from three public corpora: the two released with \cite{r11} and MAST-Data \cite{r2}, which
\texttt{get\_data.sh} fetches. No model was queried and no GPU was used. \texttt{audit.py}
re-derives every numeric assertion in this paper from the data and reports a pass or fail for
each.

\begin{footnotesize}
\begin{verbatim}
python3 anc/code/verify.py        # Theorem 1 and the tree-shape results
python3 anc/code/theory.py        # Theorems 2 and 3, the design rule, the optimal depth
python3 anc/code/measure_rb.py    # Equation (6), the aggregate exponent
python3 anc/code/diagnose.py      # Table 1 and the within-trace diagnostics
python3 anc/code/estimate.py      # Equation (7), Table 2, cluster bootstrap
python3 anc/code/measure_C5.py    # Equation (8), Table 3, the per-level constant
python3 anc/code/briefing.py      # Equations (9)-(10), Table 4, the briefing loss
python3 anc/code/production.py    # Section 7 production parameters
python3 anc/code/revealed.py      # the two biased regressions, and the inversion
python3 anc/code/hazard.py        # the identified hazard model
python3 anc/code/cost2.py         # Section 6, the calibrated token axis
python3 anc/code/figures.py       # Figures 1, 2, 3, 5, 6
python3 anc/code/fig6.py          # Figure 4
python3 anc/code/fig7.py          # Figure 7
python3 anc/code/fig8.py          # Figure 8
python3 anc/code/audit.py         # re-derives every number asserted above
python3 anc/code/verify_pdf.py    # checks the typeset PDF asserts exactly those numbers
python3 anc/code/verify_refs.py   # checks every arXiv citation against the arXiv API
\end{verbatim}
\end{footnotesize}

\section*{Authorship}

Rong He is the author of this work and of all results reported in it. Claude (Anthropic) was
used as a research assistant during development; that contribution is secondary and is
disclosed here rather than in the commit history.

\end{document}